\documentclass[draft]{birkjour}
\usepackage[noadjust]{cite}
\usepackage{xcolor}
\RequirePackage[all]{xy}

\newtheorem{thm}{Theorem}[section]

\newtheorem{lem}[thm]{Lemma}

\theoremstyle{definition}

\theoremstyle{remark}
\newtheorem{rem}[thm]{Remark}

\newtheorem*{ex}{Example}
\numberwithin{equation}{section}

\newcommand{\BibTeX}{B\kern-0.1emi\kern-0.017emb\kern-0.15em\TeX}
\newcommand{\XYpic}{$\mathrm{X\kern-0.3em\raisebox{-0.18em}{Y}}$-$\mathrm{pic}\,$}

\newcommand{\cl}{C \kern -0.1em \ell}  %%Clifford algebra

\newcommand{\ed}{\end{document}}

\begin{document}

%-------------------------------------------------------------------------
% editorial commands: to be inserted by the editorial office
%
%\firstpage{1} \volume{228} \Copyrightyear{2004} \DOI{003-0001}
%\seriesextraline{This is the Concrete Title of this Book\br H.E. R and S.T.C. W, Eds.}
%\firstpage{1}

\title[Introducing Multidimensional Dirac--Hestenes Equation]
 {Introducing Multidimensional\\ Dirac--Hestenes
 Equation}
%----------Author 1
\author[S.~Rumyantseva]{Sofia Rumyantseva}
\address{%
HSE University, 101000\\ Moscow, Russia
}
\email{srumyanceva@hse.ru}
%
% \thanks{This file has been typeset with the option \texttt{draft} to illustrate that feature and its purpose.}
%----------Author 2
\author[D.~Shirokov]{Dmitry Shirokov}
%\author[]{Rafa\l \ Ab\l amowicz}
\address{%
HSE University, 101000\\ Moscow, Russia;\\
 Institute for Information Transmission Problems\\ of the Russian Academy of Sciences, 127051\\ Moscow, Russia}
\email{dshirokov@hse.ru, shirokov@iitp.ru}
\date{\today}
\dedicatory{Last Revised:\\ \today}
%%% ----------------------------------------------------------------------
\begin{abstract}
It is easier to investigate phenomena in particle physics geometrically by exploring a real solution to the Dirac--Hestenes equation instead of a complex solution to the Dirac equation. The current research presents a formulation of  the multidimensional Dirac--Hestenes equation. Since the matrix representation of the complexified (Clifford) geometric algebra $\mathbb{C}\otimes\cl_{1,n}$ depends on the parity of $n$, we examine even and odd cases separately. In the geometric algebra $\cl_{1,3}$, there is a lemma on a unique decomposition of an element of the minimal left ideal into the product of the idempotent and an element of the real even subalgebra. The lemma is used to construct the four-dimensional Dirac--Hestenes equation. The analogous lemma is not valid in the multidimensional case, since the dimension of the real even subalgebra of $\cl_{1,n}$ is bigger than the dimension of the minimal left ideal for $n>4$. Hence, we consider the auxiliary real subalgebra of $\cl_{1,n}$ to prove a similar statement. We present the multidimensional Dirac--Hestenes equation in $\cl_{1,n}$. We prove that one might obtain a solution to the multidimensional Dirac--Hestenes equation using a solution to the multidimensional Dirac equation and vice versa. We also show that the multidimensional Dirac--Hestenes equation has gauge invariance.  
\end{abstract}
\label{page:firstblob}
%%% ----------------------------------------------------------------------
\maketitle
%%% ----------------------------------------------------------------------
%\tableofcontents

\section{Introduction}
In Minkowski space $\mathbb{R}^{1,3}$, the classical four-dimensional Dirac equation is equivalent to the Dirac--Hestenes equation \cite{lasenby1998gravity,Lounesto,Hestenes}. It means that we can obtain a solution to the Dirac--Hestenes equation using a solution to the Dirac equation, and conversely.
The Dirac--Hestenes equation gives a deeper understanding of geometry in various tasks, as the considering wave function is entirely real. 
The current research presents a formulation of  the multidimensional Dirac--Hestenes equation in the real (Clifford) geometric algebra $\cl_{1,n}$.

%Generalizating known theories to the multidimensional case can potentially lead to the development of the unified field theory. 
Generalizing known theories to the multidimensional case can potentially lead to the development of a unified field theory. Unified field theories aim to describe all fundamental forces and particles within a single coherent model, often necessitating the inclusion of additional spatial dimensions. The geometric algebra approach, with its seamless integration of spinor fields, gauge transformations, and spacetime geometry, offers a powerful tool for addressing these challenges \cite{Benn}. For instance, in M-theory, where spacetime dimensions extend beyond four, spinorial solutions must respect higher-dimensional Clifford algebras \cite{Catto}. The multidimensional Dirac--Hestenes equation can provide a foundation for exploring such extensions, potentially bridging gaps between quantum field theories and general relativity. This aligns with the goals of unified theories, such as those proposed in the context of supergravity or string compactifications, where higher-dimensional geometry plays a pivotal role.

Therefore, there is a growing interest in considering the multidimensional Dirac equation \cite{Chen,Lonigro,Jiang}. The multidimensional Dirac--Hestenes equation is expected to have characteristics similar to the multidimensional Dirac equation. Investigating the Dirac--Hestenes equation provides the advantage of a clearer understanding of geometric concepts in physics. In the classical Dirac--Hestenes equation, the imaginary unit is interpreted as a generator of rotations in the spacelike plane orthogonal to the one containing the electron current and spin vectors. Inspired by this, we use the same imaginary unit equivalent to introduce the multidimensional Dirac--Hestenes equation. For a more detailed description of the advantages of the real equation compared to the complex equation, we refer the reader to the works of Hestenes \cite{hestenes2012clifford, Hestenes, Hestenes2}.

Furthermore, generalizing the properties of the four-dimensional Dirac--Hestenes equation to the multidimensional case is essential. The classical Dirac equation \cite{Thaller} and the corresponding Dirac--Hestenes equation \cite{lasenby1998gravity,Hestenes} have gauge invariance. 
The electromagnetic potential is not uniquely determined, since it is possible to add to it any vector field with a vanishing curl without changing the physical consequences.
We get the analogous statement for the multidimensional Dirac--Hestenes equation.

The complexified geometric algebra $\mathbb{C}\otimes\cl_{1,n}$ has a unique irreducible matrix representation of minimal dimension when $n$ is odd. In the other case, there are two nonequivalent irreducible matrix representations of minimal dimension \cite{Benn,ndimension,Traubenberg}. A multidimensional Dirac spinor can be either a semi-spinor or a double spinor when $n$ is even. Therefore, we consider three cases separately.

The paper is organized as follows. In Section \ref{secMF}, we introduce the multidimensional Dirac equation in the matrix formalism, while in Section \ref{secGAF}, we discuss it in the geometric algebra formalism. In Section \ref{secDec}, we establish lemmas concerning a unique decomposition of an element of the left ideal into the product of the idempotent and an element of the auxiliary even real subalgebra of the geometric algebra $\cl_{1,n}$ in different cases: spinor, semi-spinor, and double spinor. Using these lemmas, we present the multidimensional Dirac--Hestenes equation in Section \ref{secDH} and prove its equivalence to the multidimensional Dirac equation. We show that the multidimensional Dirac--Hestenes equation has the gauge invariance in Section \ref{secInv}. The conclusions follow in Section \ref{secCon}.

This work is an extended version of the short note \cite{Rum} in Conference Proceedings (Empowering Novel Geometric Algebra for Graphics \& Engineering Workshop within
the International Conference Computer Graphics International 2024). In the short note, the case $n=2d-1$ has been examined. In this paper, both cases $n=2d-1$ and $n=2d$ are addressed. Sections \ref{secDec}, \ref{secDH}, and 
\ref{secInv} are extended; Lemmas \ref{lemmasemi}--\ref{decompositiondouble} and Theorems \ref{OddD-Hsemi}--\ref{OddD-Hdouble} are presented for the first time.

% ------------------------------------------------------------------------

\section{Multidimensional Dirac equation in matrix formalism}
\label{secMF}

We consider pseudo-Euclidean space $\mathbb{R}^{1,n}$ with Cartesian coordinates $\{x^{\mu}\}_{\mu=0}^n$. Partial derivatives are denoted by $\partial_\mu = \partial/\partial x^{\mu}$.

The metric tensor $\mathbb{R}^{1,n}$ is given by the diagonal matrix $\eta$, where the first element is $1$ and the remaining elements on the main diagonal are $-1$:
\begin{equation}
    \label{diagonalmatrix}
    \eta=(\eta^{\mu\nu})_{\mu,\nu=0}^n=\operatorname{diag}(1,-1,-1,\ldots,-1).
\end{equation}

Let us consider the multidimensional Dirac equation in the matrix formalism.
We denote the mass of a particle by $m$. For convenience, we assume that Planck constant, the charge of a particle, and the speed of light are equal to $1$.
The electromagnetic vector-potential $\textbf{a}(x)$ depends on a point $x$ of the pseudo--Euclidean space $\mathbb{R}^{1,n}$. In other words, $\textbf{a}(x)=(a_0(x),\ldots,a_n(x)):\mathbb{R}^{1,n}\to\mathbb{R}^{n+1}$. A solution to the Dirac equation is a complex-valued vector function $\varphi(x):\mathbb{R}^{1,n}\to\mathbb{C}^{2^{[(n+1)/2]}}$, where $[n]$ is the integer part of $n$. In the literature,  $\varphi(x)$ is called a wave function or a Dirac spinor. The multidimensional Dirac equation has the form \cite{Chen,Jiang}:
\begin{equation}
    \label{direq}
    \sum_{\mu=0}^{n}{i \gamma^{\mu}(\partial_{\mu}+i  a_{\mu}(x))\varphi(x)=m\varphi(x)},
\end{equation}
where $i$ is imaginary unit and matrices $\{\gamma^{\mu}\}_{\mu=0}^n$ satisfy the following anticommutation relations:
\begin{equation}
    \label{gamma}
    \gamma^{\mu} \gamma^{\nu}+\gamma^{\nu} \gamma^{\mu} =2\eta^{\mu\nu}\mathbb{I}, \quad  \mu,\nu\in\{0,1,\ldots,n\}, 
\end{equation}
where $\mathbb{I}$ is the identity matrix of size $2^{[(n+1)/2]}$. Note that the matrices $\{\gamma^{\mu}\}$ are the Dirac gamma matrices in the special case $n=3$~\cite{Dirac}.

 It is worth noting that unitarily equivalence of the Dirac equation depends on the parity of $n$~\cite{kalf2001essential}. In the case of an odd $n$, for two sets of the matrices $\{\gamma^{\mu}\}$ that satisfy relation \eqref{gamma} and are related by a unitary transformation, the corresponding Dirac equations are also related by the same unitary transformation. However, in the case $n=2d$, the Dirac equation is unitarily equivalent to two Dirac equations. To describe these two equations, we can use block diagonal matrices of size $2^{d+1}$ instead of matrices of size $2^{d}$ as the gamma matrices.

It is evident that the multidimensional Dirac equation has a gauge invariance. The statement can be proven analogously to the approach employed in the four-dimensional case.
An electromagnetic potential $a_{\mu}(x)$ can be replaced by $\Tilde{a}_{\mu}(x) = a_{\mu}(x)-\partial_{\mu} G(x)$, where $e^{i G(x)}$ takes value in $U(1)$. Therefore, if $\varphi(x)$ is a solution to equation~\eqref{direq} with an electromagnetic potential $a_{\mu}(x)$, then $\Tilde{\varphi}(x) = e^{i G(x)}\varphi(x)$ is a solution to equation \eqref{direq} with a shifted electromagnetic potential $\Tilde{a}_{\mu}(x)$. In the paper, we show that the multidimensional Dirac--Hestenes equation also has gauge invariance.

% ------------------------------------------------------------------------

\section{Multidimensional Dirac equation in geometric algebra formalism}
\label{secGAF}

Using a geometric algebra is one of the ways to investigate issues in modern mathematical physics \cite{hestenes2012clifford,Hestenes2,Lasenby}. We consider the real geometric algebra $\cl_{1,n}$. The generators $e^0,e^1,\ldots,e^n$ satisfy the following anticommutation relations:
\begin{equation}
    \label{generator}
    e^{\mu} e^{\nu}+e^{\nu} e^{\mu} =2\eta^{\mu\nu}e, \quad  \mu,\nu\in\{0,1,\ldots,n\},
\end{equation}
where $\eta$ is diagonal matrix \eqref{diagonalmatrix} and $e$ is the identity element. Note that relations (\ref{gamma}) and (\ref{generator}) are similar. It means that the generators of the geometric algebra satisfy the same anticommutation relations as the matrices $\{\gamma^{\mu}\}$. Therefore, it becomes possible to consider the multidimensional Dirac equation not only in the matrix formalism but also in the geometric algebra formalism.

The basis of considering geometric algebra $\cl_{1,n}$ consists of all possible ordered products of the generators: 
\begin{eqnarray}
e^{\mu_1}e^{\mu_2}\cdots e^{\mu_k}=e^{\mu_1 \mu_2\ldots \mu_k},\quad 0 \leq \mu_1<\mu_2<\cdots<\mu_k\leq n.
\end{eqnarray}

Hence, the basis decomposition of an element $U\in\cl_{1,n}$ (it is also called a multivector) is:
\begin{eqnarray}
    U = u e +\sum_{\mu=0}^n u_{\mu} e^{\mu} +\sum_{\mu,\nu=0,\mu< \nu}^n u_{\mu\nu} e^{\mu\nu}+\cdots+u_{01\ldots n} e^{01\ldots n},
\end{eqnarray}
where $u,u_{\mu},u_{\mu\nu},\ldots,u_{01\ldots n}$ are real scalars. The basis decomposition can be rewritten using multi-indices:
\begin{equation}
\label{U=sum}
    U = \sum_{M} u_{M} e^M, \quad u_M\in\mathbb{R},
\end{equation}
where $M= \mu_1 \mu_2\ldots \mu_k$. We denote the length of multi-index $M$ by $|M|=k$, $k=0,1,\ldots,n+1$. A multi-index is called even if its length is even, and odd otherwise. When the length of a multi-index $M$ is zero, the corresponding $e^M$ represents the identity element $e$. In contrast, if $M=0$ as an individual index, then its length is considered to be $1$ and $e^0$ acts as a generator of the algebra. Note, the dimension of $\cl_{1, n}$ is $2^{n+1}$.

We denote a subalgebra of the geometric algebra $\cl_{1,n}$ constructed on selected generators. For instance, if a subalgebra is constructed only on the generators with even indices, we denote it by:
\begin{eqnarray}
\cl(e^0,e^2,e^4,e^6,\cdots)	\subset \cl_{1,n}.
\end{eqnarray}

Let $\cl_{1,n}^{(0)}$ be an even subalgebra of $\cl_{1,n}$ that is a linear span of the basis elements with even multi-indices:
\begin{eqnarray}
    \cl_{1,n}^{(0)}=\{U\in \cl_{1,n}| U = \sum_{|M|=2k} u_M e^M\}, \quad \operatorname{dim}\cl_{1,n}^{(0)} = 2^{n}.
\end{eqnarray}
An element of the even subalgebra $\cl_{1,n}^{(0)}$  is called an even element.

Let $\cl_{1,n}^{(1)}$ be an odd subspace of $\cl_{1,n}$. The basis of the odd subspace $\cl_{1,n}^{(1)}$ consists of the basis elements of $\cl_{1,n}$ with odd multi-indices:
\begin{eqnarray}
    \cl_{1,n}^{(1)}=\{U\in \cl_{1,n}| U = \sum_{|M|=2k-1} u_M e^M\}, \quad \operatorname{dim}\cl_{1,n}^{(1)} = 2^{n}.
\end{eqnarray}
An element of the odd subspace $\cl_{1,n}^{(1)}$  is called an odd element. 

To describe a basis of a space to which a Dirac spinor belongs, we introduce a projection operation onto the subspace of zero grade. If a multivector has decomposition \eqref{U=sum}, then the operation has the form:
\begin{eqnarray}
    \langle U \rangle = u.
\end{eqnarray}

In the paper, we also consider the complexified geometric algebra $\mathbb{C}\otimes\cl_{1,n}$. The basis decomposition of $U\in\mathbb{C}\otimes\cl_{1,n}$ is similar to decomposition (\ref{U=sum}) but the constants $\{u_M\}$ are complex scalars.

Let us introduce an operation of Hermitian conjugation in complexified geometric algebras. Actually, its definition depends on the signature. For the signature $(1,n)$, the operation has the form \cite{ndimension,Shirokov}:
\begin{eqnarray}
    U^{\dagger}=e^0 U^* e^0,
\end{eqnarray}
where the star denotes the superposition of reversion and complex conjugation:
\begin{eqnarray}
    U^* = \sum_{M} (-1)^{\frac{|M|(|M|-1)}{2}}\bar{u}_{M} e^M, \quad u_M\in\mathbb{C}.
\end{eqnarray}

The complexified geometric algebra $\mathbb{C}\otimes\cl_{1,n}$ is inner product space with the inner product:
\begin{equation}
\label{inner}
    (U,V) = \langle U^{\dagger} V\rangle,\quad U,V\in \mathbb{C}\otimes\cl_{1,n}.
\end{equation}

We consider the Hermitian idempotent $t$:
\begin{eqnarray}
   t^2=t,\quad t^{\dagger}=t,
\end{eqnarray}
and the corresponding left ideal $L(t)$ generated by $t$:
\begin{eqnarray}
    L(t) = \{U\in \mathbb{C} \otimes \cl_{1,n}| U t = U\}.
\end{eqnarray}

If a left ideal $L(t)$ does not contain another left ideal except itself and $L(0)$, then it is called a minimal left ideal. The corresponding idempotent $t$ is called a primitive idempotent.

We consider the Dirac equation in the geometric algebra formalism. It is convenient to investigate a multidimensional Dirac spinor as an element of the left ideal $L(t)$ \cite{Lounesto,Benn,Riesz}: 
\begin{eqnarray}
   \varphi(x): \mathbb{R}^{1,n} \to L(t). 
\end{eqnarray}

Actually, there is a difference between cases $n=2d-1$ and $n=2d$ due to the isomorphism between complexified geometric algebras and matrix algebras \cite{ndimension,Sobczyk}:
\begin{equation}
    \mathbb{C}\otimes\cl_{1,2d-1}\simeq \operatorname{Mat}(2^{d},\mathbb{C}),
\end{equation}
\begin{equation}
\mathbb{C}\otimes\cl_{1,2d}\simeq \operatorname{Mat}(2^{d},\mathbb{C})\oplus \operatorname{Mat}(2^{d},\mathbb{C}).
\end{equation}
In the first case, a Dirac spinor belongs to a minimal left ideal, while in the latter case, a Dirac spinor might belong not only to a minimal left ideal. There are two types of Dirac spinor for $n=2d$: a semi-spinor and a double spinor. In the matrix representation, the semi-spinor consists of two block matrices --- a matrix with one non-zero column and a zero matrix. The double spinor also consists of two block matrices, each of which belongs to a minimal left ideal. 

The multidimensional Dirac equation with an electromagnetic vector-potential $\textbf{a}(x)$ in the geometric algebra formalism is:
\begin{equation}
\label{Dirac}
    \sum_{\mu=0}^{n}{i e^{\mu}(\partial_{\mu}+i  a_{\mu}(x))\varphi(x)=m\varphi(x)}.
\end{equation}

Typically, a spinor is understood as a complex column vector. In the formalism of geometric algebra, however, a spinor is represented as an element of a minimal left ideal. There is a direct connection between these two approaches: in the matrix representation, an element of the minimal left ideal corresponds to a matrix with only one nonzero column.

% ------------------------------------------------------------------------

\section{Decomposition of an element of the left ideal}
\label{secDec}

Initially, we remind several facts for the special case $n=3$. 
It is known that an element of the minimal left ideal $L(t)$ has a unique decomposition into the product of an element of the even real subalgebra $\cl^{(0)}_{1,3}$ and the corresponding idempotent $t$~\cite{Lounesto,marchuk2012field}.

\begin{lem}
\label{n=3}
Let $L(t)$ be the minimal left ideal generated by the idempotent $t$:
\begin{eqnarray}
    t = \frac{1}{4}(e+e^0)(e+i e^{12})\in \mathbb{C}\otimes\cl_{1,3}.
\end{eqnarray}
Then there is the unique decomposition:
\begin{eqnarray}
    \forall \varphi \in L(t) \quad \exists! \Psi \in \cl^{(0)}_{1,3}:\quad \varphi = \Psi t.
\end{eqnarray}
\end{lem}

Using Lemma \ref{n=3}, it can be figured out that it is possible to obtain a solution $\Psi(x)$ to the Dirac--Hestenes equation from a solution $\psi(x)$ to the Dirac equation, and conversely \cite{Hestenes,Hestenes2}. The Dirac--Hestenes equation has the form:
\begin{eqnarray}
    \sum_{\mu = 0}^3 e^{\mu} (\partial_{\mu} \Psi(x)  - \Psi(x) a_{\mu}(x) e^{12}) e^0 = m \Psi(x) e^{12},\quad \Psi(x)\in \cl^{(0)}_{1,3}.
\end{eqnarray}

Actually, the matrix representation of $\mathbb{C}\otimes\cl_{1,n}$ depends on the parity of $n$. In the case $n=2d-1$, there is only one irreducible matrix representation of minimal dimension. In the other case $n=2d$, there are two non-equivalent irreducible matrix representations of minimal dimension. Therefore, we analyze these two cases separately.

\subsection{The case $n=2d-1$}

In this subsection, we introduce Lemma \ref{decomposition} that is a generalization of Lemma~\ref{n=3} to the multidimensional case $n=2d-1$.

The real dimensions of the left ideal $L(t)$ and the even real subalgebra $\cl^{(0)}_{1,3}$ are the same: 
\begin{eqnarray}
   \operatorname{dim}_{\mathbb{R}} L(t) = \operatorname{dim}_{\mathbb{R}}\cl^{(0)}_{1,3}=8. 
\end{eqnarray}
However, the equality of the dimensions does not hold for $d>2$ since:
\begin{eqnarray}
   \operatorname{dim}_{\mathbb{R}} L(t) = 2^{d+1}, \quad \operatorname{dim}_{\mathbb{R}} \cl^{(0)}_{1,2d-1} = 2^{2d-1}. 
\end{eqnarray}
Hence, if we replace $\cl_{1,3}$ by $\cl_{1,2d-1}$, then the similar statement to Lemma~\ref{n=3} will not be valid for the case $d>2$.
Therefore, it is necessary to introduce another real subalgebra to which $\Psi(x)$ belongs. This subalgebra has a smaller dimension than the even real subalgebra $\cl^{(0)}_{1, 2d-1}$.

One of the ways to fix the primitive Hermitian idempotent $t\in  \mathbb{C} \otimes \cl_{1,2d-1}$ is \cite{Shirokov}:
\begin{equation}
\label{event}
    t = \frac{1}{2}(e+ e^0)\prod\limits_{\mu = 1}^{d-1}\frac{1}{2}(e+i e^{2\mu-1}e^{2\mu}) \in  \mathbb{C} \otimes \cl_{1,2d-1}.
\end{equation}
In this subsection, we consider the minimal left ideal $L(t)$ generated by (\ref{event}).

An important benefit of applying geometric algebra is the ability to interpret geometric results when the Dirac--Hestenes equation is analyzed. We use the fixed multivector $I$ for reducing the Dirac equation to a form where all values are real:
\begin{equation}
    \label{it}
    I:= -e^{12},\quad it= It=tI.
\end{equation}

To introduce the multidimensional Dirac--Hestenes equation where a wave function $\Psi(x)$ consists only of basis elements with even multi-indices, we also use another fixed multivector $E$:
\begin{equation}
    \label{et}
    E:= e^0,\quad t= Et=tE.
\end{equation}

For an explicit description of the minimal left ideal $L(t)$, it is convenient to consider the auxiliary algebra $Q$ which is generated by generators with odd indices:
\begin{eqnarray}
    Q = \cl(e^1,e^3,\ldots,e^{2d-3},e^{2d-1})\subset \mathbb{C}\otimes\cl_{1,2d-1}, \quad \operatorname{dim}_{\mathbb{C}} Q = 2^{d}.
\end{eqnarray}
The corresponding basis elements, which are ordered products of the generators $e^1,e^3,\cdots,e^{2d-3},e^{2d-1}$, are denoted by $\{c^\mu\}_{\mu=1}^{2^d}$. 
Therefore, the orthonormal basis of the left ideal $L(t)$ has the form \cite{Shirokov}:
\begin{eqnarray}
    \tau^{k}=(\sqrt{2})^{d} c^{k} t, \quad k = 1,2,\ldots, 2^{d},
\end{eqnarray}
\begin{eqnarray}
    (\tau^{j},\tau^{k})=\delta^{jk}, \quad  j= 1,2,\ldots, 2^{d},
\end{eqnarray}
where $\delta^{\mu\nu}$ is Kronecker delta and the parentheses denote inner product \eqref{inner}.

Since a solution to the Dirac--Hestenes equation should be an even element that belongs to a real geometric algebra, we consider the following real subalgebra~$Q'$:
\begin{eqnarray}
    Q'=\cl(e^0,e^1,e^2,e^3,e^5,\ldots,e^{2d-3},e^{2d-1})\subset \cl_{1,2d-1}, \operatorname{dim}_{\mathbb{R}} Q'= 2^{d+2}.
\end{eqnarray}
It means that its generators are $e^0$, $e^2$, and generators with odd indices. Note that we get the algebra $Q' = \cl_{1,3}$ for $d=2$ and:
\begin{eqnarray}
    \operatorname{dim}_{\mathbb{R}} L(t) = \operatorname{dim}_{\mathbb{R}} Q'^{(0)} = 2^{d+1}.
\end{eqnarray}

In Lemma \ref{decomposition} for the multidimensional case $n=2d-1$, which is a generalization of Lemma \ref{n=3}, the even subalgebra $Q'^{(0)}$ is used instead of the even subalgebra $\cl_{1,3}^{(0)}$. 
Lemma \ref{Yt=0} is used to prove the uniqueness of the decomposition in Lemma~\ref{decomposition} and to construct the multidimensional Dirac--Hestenes equation.

\begin{lem}
\label{Yt=0}
    Let $Q'$ be $\cl(e^0,e^1,e^2,e^3,e^5,e^7,\ldots,e^{2d-1})$ and $t$ have form~\eqref{event}. If $Y\in~Q'^{(0)}$ and $Yt=0$, then $Y=0$.
\end{lem}

\begin{rem}
  If we replace $Q'^{(0)}$ by $\cl^{(0)}_{1,2d-1}$ in the statement of Lemma \ref{Yt=0}, then the new statement will not be valid. Let us present an example for the case $d=3$. If $Y = e^{12}-e^{34} \in \cl^{(0)}_{1,5}$, then we get $Yt = -it+it=0$.
\end{rem}

\begin{proof}
An element $c^{k}$ is an odd or even basis element of the algebra $Q$. Since $Y$ belongs to $Q'^{(0)}$, it is convenient to consider auxiliary elements $\{f^{k}\}$ instead of $\{c^{k}\}$. The elements $\{f^{k}\}$ belong to $Q'^{(0)}$:
  \begin{equation}
    \label{Fk}
    f^{k} = \begin{cases}
    c^{k},\quad \text{if $c^{k}$ is even},\\
    c^{k} e^0,\quad \text{if $c^{k}$ is odd}.
\end{cases}
\end{equation}

Using property (\ref{et}), we rewrite the basis $\{\tau^{k}\}$ of $L(t)$ in the following form:
\begin{equation}
\label{tauF}
    \tau^{k}=(\sqrt{2})^{d} f^{k} t, \quad k = 1,2,\ldots, 2^{d}.
\end{equation}
    
     Let us decompose $Y$ into four sums: the first sum contains neither $e^1$ nor $e^2$; the second sum, conversely, contains $e^{1}$ and $e^2$; the third sum contains $e^1$ and does not contain $e^2$; the fourth sum, conversely, contains $e^2$ and does not contain~$e^1$. 
     
     Let $\{h^{k}\}$ be a subset of $\{f^{k}\}$ such that $h^{k}$ does not contain the generator $e^1$. From the construction of $\{f^{k}\}$ it follows that $\{h^{k}\}$ are even elements and do not contain the generator $e^2$. 
    Also, let $\{g^{k}\}$ be odd basis elements of the algebra $Q'$ that do not contain the generators $e^1$ and $e^2$. Thus, an element $Y\in Q'^{(0)}$ has the following basis decomposition:
    \[
    Y = \sum_{k=1}^{2^{d-1}}\left( y_{k} h^{k} + y_{k12}h^{k} e^{12}+ y_{k1} g^{k} e^{1}+  y_{k2}g^{k} e^{2}\right),\quad y_{k},\,y_{k12},\,y_{k1},\,y_{k2}\in\mathbb{R}.
    \]

Let us multiply both sides of this equality on the right by primitive Hermitian idempotent (\ref{event}). Taking into account property (\ref{it}) and the fact $e^1 e^1 = -e$, we get:
    \begin{eqnarray*}
&Yt &= \sum_{k=1}^{2^{d-1}}\left( y_{k} h^{k} t - y_{k12}h^{k} i t + y_{k1} g^{k} e^{1} t- y_{k2}g^{k} e^1e^{12} t\right)=\\
      &  & = \sum_{k=1}^{2^{d-1}} \left((y_{k} - i y_{k12})h^{k} t +  (y_{k1}+i y_{k2})g^{k} e^{1} t\right).
\end{eqnarray*}

Actually, the union of $\{h^{k}\}$ and $\{g^{k} e^1\}$ is the set $\{f^{k}\}$ up to sign. Using~\eqref{tauF}, we obtain:
$$ Yt =  \sum_{k=1}^{2^{d}} (\alpha_{k}+i \beta_{k}) \tau^{k}, \quad \alpha_{k},\,\beta_{k}\in\mathbb{R}.$$
If $Y t =0$, we get $\alpha_{k}=\beta_{k}=0$. Hence, $Y=0$.

\end{proof}

\begin{lem}
\label{decomposition}
Let $Q'$ be $\cl(e^0,e^1,e^2,e^3,e^5,e^7,\ldots,e^{2d-1})$ and $L(t)$ be the minimal left ideal generated by idempotent $t$ \eqref{event}. Then:
\begin{equation}
\label{varphi}
   \forall \varphi \in L(t)\quad \exists! \Psi\in Q'^{(0)}:\quad \varphi = \Psi t.
\end{equation}
\end{lem}

\begin{proof}
First, we show the existence of decomposition \eqref{varphi}. The basis decomposition of $\varphi$ is:
\begin{equation*}
    \varphi= \sum_{k=1}^{2^{d}}(\alpha_{k} + i \beta_{k}) \tau^k= (\sqrt{2})^{d}\sum_{k=1}^{2^{d}}(\alpha_{k} + i \beta_{k}) f^{k} t,
\end{equation*}
where $\alpha_{k},\, \beta_{k}\in\mathbb{R}$, and $f^{k}$ is defined in \eqref{Fk}. Using property (\ref{it}), we get:
 \[\varphi =  (\sqrt{2})^{d}\left(\sum_{k=1}^{2^{d}} f^{k}(\alpha_{k} + I\beta_{k})\right) t \Rightarrow \Psi  = (\sqrt{2})^{d}\sum_{k=1}^{2^{d}} f^{k}(\alpha_{k} + I\beta_{k})\in Q'^{(0)}.\]

    We prove the uniqueness of decomposition \eqref{varphi} by contradiction. Let us assume that the decomposition of $\varphi\in L(t)$ is not unique:
    \[
    \exists \Psi_1,\Psi_2\in Q'^{(0)}:\quad    \Psi_1\neq \Psi_2 \text{ and }\, \Psi_1 t = \varphi,\,\Psi_2 t =\varphi.
    \]
    Subtracting the equations yields $(\Psi_1-\Psi_2)t=0$.
    It follows from Lemma \ref{Yt=0} that $\Psi_1=\Psi_2$. 
\end{proof}

The uniqueness of the decomposition \eqref{varphi} is used to prove the equivalence between the Dirac equation and the Dirac--Hestenes equation in Section~\ref{secDH}.

\subsection{The case $n=2d$}

In this subsection, we introduce generalizations of Lemma \ref{n=3} to the multidimensional case of a semi-spinor and a double spinor. 
We construct idempotents in a similar way as for the case $n=2d-1$.

\begin{lem}
\label{lemmasemi}
    Let $t$ have the form:
    \begin{equation}
    \label{oddtsemi}
    t = \frac{1}{2}(e+ e^0)\prod\limits_{\mu = 1}^{d}\frac{1}{2}(e+i e^{2\mu-1}e^{2\mu}) \in  \mathbb{C} \otimes \cl_{1,2d}.
    \end{equation}
    Then $t$ is a primitive
Hermitian idempotent and $L(t)$ is a space to which a semi-spinor belongs.
\end{lem}

\begin{proof}
    If we prove the statement that $t$ is a primitive
Hermitian idempotent, then $L(t)$ is a space to which a semi-spinor belongs according to the definition of semi-spinor. Therefore, we should prove only the first statement. 

We show that $t$ is a Hermitian idempotent. Note that all terms in product \eqref{oddtsemi} are commute. Therefore, we can consider all terms separately. We get:
$$\frac{1}{4}(e+e^0)^2 = \frac{1}{2}(e+e^0) = \frac{1}{2}(e+e^0)^{\dagger},$$
$$\frac{1}{4}(e+i e^{2\mu-1}e^{2\mu})^2 = \frac{1}{2}(e+i e^{2\mu-1}e^{2\mu}) = \frac{1}{2}(e+i e^{2\mu-1}e^{2\mu})^{\dagger}.$$
It is notable that properties \eqref{it} and \eqref{et} hold for idempotent \eqref{oddtsemi}.

If $L(t)$ has the minimum possible dimension, then the idempotent $t$ is primitive. We show the fact:
$$\operatorname{dim}_{\mathbb{C}} L(t) = 2^{d}.$$

To construct the basis $\{\tau^{k}\}_{k=1}^{2^d}$ of $L(t)$, we fix the auxiliary algebra $Q$ as for the case $n=2d-1$:
$$Q = \cl(e^1,e^3,\ldots,e^{2d-3},e^{2d-1})\subset \mathbb{C}\otimes\cl_{1,2d}, \quad \operatorname{dim}_{\mathbb{C}} Q = 2^{d}.$$
Thus, the algebra $Q$ contains the generators with odd indices. Let $c^{k}$ be a basis element of $Q$. 
We denote $f^{k}$ by \eqref{Fk} and $\tau^{k}$ by \eqref{tauF}.

Let us consider the inner product of two elements of the set $\{\tau^{k}\}$:
\begin{equation*}
    (\tau^{j},\tau^{k}) = 2^{d+1}\langle(f^{j} t)^{\dagger} f^{k} t\rangle= 2^{d+1}\langle c^{j\,\dagger} c^{k} t \rangle,\quad j=1,2,\ldots,2^d.
\end{equation*}

If parentheses in \eqref{oddtsemi} are expanded, it is noticeable that all terms contain an even generator. Therefore, $t$ does not contain terms consisting of $c^{k}$. According to the construction of $c^{k}$, we obtain:
\begin{equation*}
  2^{d+1}  \langle c^{j\,\dagger} c^{k} t \rangle = \langle c^{j\,\dagger} c^{k} \rangle = \delta^{jk}.
\end{equation*}
Hence, the set $\{\tau^{k}\}$ is the orthonormal basis of $L(t)$.
\end{proof}

\begin{rem}
We can prove Lemma \ref{lemmasemi} via the matrix representation. 
    Taking into account the recursive method of construction the matrix representation \cite{shirokov2021computing}, it becomes clear that the matrix representations of generators from $\mathbb{C} \otimes \cl_{1,2d}$ and $\mathbb{C} \otimes \cl_{1,2d+1}$ are the same, except the generator $e^{2d+1}$. Therefore, idempotent \eqref{event} belonging to $\mathbb{C} \otimes \cl_{1,2d+1}$ has the same matrix representation as idempotent~\eqref{oddtsemi} belonging to $\mathbb{C} \otimes \cl_{1,2d}$. It follows that $t$ is a primitive
Hermitian idempotent.
\end{rem}

Let us introduce Lemma \ref{decompositionsemi} on the unique decomposition of a semi-spinor.

\begin{lem}
\label{Yt=0semi}
    Let $Q'$ be $\cl(e^0,e^1,e^2,e^3,e^5,e^7,\ldots,e^{2d-1})$ and $t$ have form~\eqref{oddtsemi}. If $Y\in~Q'^{(0)}$ and $Yt=0$, then $Y=0$.
\end{lem}

\begin{lem}
\label{decompositionsemi}
Let $Q'$ be $\cl(e^0,e^1,e^2,e^3,e^5,e^7,\ldots,e^{2d-1})$ and $L(t)$ be the minimal left ideal generated by idempotent $t$ \eqref{oddtsemi}. Then:
\begin{equation}
\label{varphisemi}
   \forall \varphi \in L(t)\quad \exists! \Psi\in Q'^{(0)}:\quad \varphi = \Psi t.
\end{equation}
\end{lem}

\begin{proof}
   Similarly to the previous subsection, it can be proved that Lemma~\ref{Yt=0} and Lemma~\ref{decomposition} are also valid if the idempotent has form \eqref{oddtsemi}, since the auxiliary algebra $Q$ and the basis of $L(t)$ are the same. 
\end{proof}

Before introducing the generalization of Lemma \ref{n=3} to the multidimensional case of a double spinor, we present a form of the corresponding idempotent. 

\begin{lem}
\label{lemmadouble}
    Let $t$ have the form:
    \begin{equation}
    \label{oddtdouble}
    t = \frac{1}{2}(e+ e^0)\prod\limits_{\mu = 1}^{d-1}\frac{1}{2}(e+i e^{2\mu-1}e^{2\mu}) \in  \mathbb{C} \otimes \cl_{1,2d}.
    \end{equation}
    Then $t$ is a 
Hermitian idempotent and $L(t)$ is a space to which a double spinor belongs.
\end{lem}

\begin{proof}
Let us fix the auxiliary algebra $Q$ as follows:
$$Q = \cl(e^1,e^3,\ldots,e^{2d-3},e^{2d-1}, e^{2d})\subset \mathbb{C}\otimes\cl_{1,2d}, \quad \operatorname{dim}_{\mathbb{C}} Q = 2^{d+1}.$$
In other words, the generators of $Q$ are generators with odd indices, as in the previous cases, and the generator $e^{2d}$. 

Let $c^{k}$ be a basis element of $Q$ and $\tau^{k}$ have the form:
$$\tau^{k} = (\sqrt{2})^{d} c^{k} t,\quad k=1,2,\ldots,2^{d+1}.$$

It has been shown in Theorem 8 in  \cite{Shirokov}, that element \eqref{oddtdouble} is a Hermitian idempotent and the set $\{\tau^{k}\}$ is the orthonormal basis of $L(t)$. 

A double spinor belongs to a reducible space consisting of two irreducible spaces:
\begin{equation*}
    L(t)=L(t_1)\oplus L(t_2),
\end{equation*}
where $t_1$ and $t_2$ are primitive idempotents. We show that $L(t)$ is a space to which a double spinor belongs. We denote by $\{p^{k}\}$ a subset of $\{c^{k}\}$ such that $p^{k}$ does not contain the generator $e^{2d}$. Hence, the basis decomposition of $\varphi\in L(t)$ is:
\begin{equation*}
    \begin{aligned}
        \varphi&=  (\sqrt{2})^{d}\sum_{k=1}^{2^{d+1}}(\alpha_{k} + i \beta_{k}) c^{k} t \\
&= (\sqrt{2})^{d}\sum_{k=1}^{2^{d}}(\alpha^1_{k} + i \beta^1_{k}) p^{k} t+(\sqrt{2})^{d}\sum_{k=1}^{2^{d}}(\alpha^2_{k} + i \beta^2_{k}) p^{k} e^{2d} t,
    \end{aligned}
\end{equation*}
where the first sum belongs to the minimal left ideal $L(t_1)$, since the element $t_1$ is  primitive Hermitian idempotent \eqref{event} and $L(t_1)\subset \mathbb{C}\otimes\cl_{1,2d-1}$.

We introduce $t_2$ as follows:
$$t_2 = \frac{1}{2}(e- e^0)\prod\limits_{\mu = 1}^{d-1}\frac{1}{2}(e+i e^{2\mu-1}e^{2\mu}) \in  \mathbb{C} \otimes \cl_{1,2d-1}.$$
As in the proof of Lemma \ref{lemmasemi}, it can be shown that the element $t_2$ is a primitive Hermitian idempotent and $\{(\sqrt{2})^{d} p^{k}t_2\}_{k=1}^{2^d}$ is an orthonormal  basis of $L(t_2)\subset \mathbb{C}\otimes\cl_{1,2d-1}$. 

It is notable the following property:
$$e^{2d}t = t_2 e^{2d}.$$
Therefore, we get:
$$\varphi=\varphi_1+\varphi_2 e^{2d},\quad \varphi_1\in L(t_1),\,\varphi_2\in L(t_2).$$

It remains to show that the intersection of $L(t_1)$ and $L(t_2)$ is a zero element. If we multiply an element of $L(t_1)$ by $t_2$ on the left, we do not get the same element, except the case where the element is zero. Let us consider a basis element $(\sqrt{2})^{d} p^{k}t_1$ of $L(t_1)$:
$$(\sqrt{2})^{d} p^{k}t_1 t_2 = (\sqrt{2})^{d} p^{k}\frac{1}{4}(e+ e^0)(e- e^0)\prod\limits_{\mu = 1}^{d-1}\frac{1}{2}(e+i e^{2\mu-1}e^{2\mu})=0.$$
Similarly, If we multiply an element of $L(t_2)$ by $t_1$ on the left, the result is also zero.
\end{proof}

\begin{rem}
    Lemma \ref{lemmadouble} can also be proven via the matrix representation. The matrix representation of $\mathbb{C} \otimes \cl_{1,2d}$ is constructed using the matrix representation of $\mathbb{C} \otimes \cl_{1,2d-1}$. In the previous subsection, it has been shown that the element $t_1\in \mathbb{C} \otimes \cl_{1,2d-1}$ is the primitive Hermitian idempotent. The matrix representation of $t_2\in\mathbb{C} \otimes \cl_{1,2d-1}$ is a block matrix with only one identity element lying on the diagonal. Therefore, $t_2$ is the primitive Hermitian idempotent. Finally, the matrix representation of idempotent \eqref{oddtdouble} consists of two block matrices, the matrix representations of $t_1$ and $t_2$.
\end{rem}

Let us introduce Lemma \ref{decompositiondouble} on the unique decomposition of a double spinor.

\begin{lem}
\label{Yt=0double}
    Let $Q'$ be $\cl(e^0,e^1,e^2,e^3,e^5,e^7,\ldots,e^{2d-1},e^{2d})$ and $t$ have form~\eqref{oddtdouble}. If $Y\in~Q'^{(0)}$ and $Yt=0$, then $Y=0$.
\end{lem}

\begin{lem}
\label{decompositiondouble}
Let $Q'$ be $\cl(e^0,e^1,e^2,e^3,e^5,e^7,\ldots,e^{2d-1},e^{2d})$ and $L(t)$ be the minimal left ideal generated by idempotent $t$ \eqref{oddtdouble}. Then:
\begin{equation}
\label{varphidouble}
   \forall \varphi \in L(t)\quad \exists! \Psi\in Q'^{(0)}:\quad \varphi = \Psi t.
\end{equation}
\end{lem}

\begin{proof}
 Even though the algebra $Q'$ contains the generator $e^{2d}$, idempotent~\eqref{oddtdouble} does not contain it. Therefore, the proofs are similar to the proofs of Lemma \ref{Yt=0} and Lemma \ref{decomposition}. 
\end{proof}

\section{Multidimensional Dirac--Hestenes equation}
\label{secDH}

In the previous section, we have described a way to construct the real subalgebra $Q'^{(0)}$, which contains  solutions to the multidimensional Dirac--Hestenes equation. 

Let us present the multidimensional Dirac--Hestenes equation in the case $n=2d-1$:
\begin{equation}
    \begin{aligned}
        \label{Dir-Hest}
   & \sum_{\mu = 0,1,2,3,5,7,\ldots,2d-1}e^{\mu} (\partial_{\mu} \Psi+ \Psi a_{\mu} I) E \\
    &+\sum_{\mu = 3,5,\ldots,2d-3} (\partial_{\mu+1} \Psi + \Psi a_{\mu+1} I ) e^{\mu} E I +m\Psi I=0,
    \end{aligned}
\end{equation}
where $I=-e^{12}$ and $E=e^0$. Note that the index of the first summation is odd or equal to $0$, $2$. The sums are similar. However, the generator $e^{\mu}$ appears on the left in the first sum and on the right in the second sum. According to Lemma \ref{decomposition} and relation \eqref{varphi}, combining these sums into one is not possible as $\Psi$ may not commute with the generator $e^{\mu}$.

\begin{thm}
\label{EvenD-H}
    Let $t$ be primitive Hermitian idempotent (\ref{event}) and $Q'$ be the real algebra:
    \begin{equation}
        Q' = \cl(e^0,e^1,e^2,e^3,e^5,e^7,\ldots,e^{2d-1})\subset \cl_{1,2d-1}.\label{Q1}
    \end{equation}

    If  $\varphi(x)\in L(t)$ is a solution to multidimensional Dirac equation (\ref{Dirac}), then $\Psi(x)\in Q'^{(0)}: \varphi(x)=\Psi(x)t$ is the corresponding solution to multidimensional Dirac--Hestenes equation \eqref{Dir-Hest}.

If $\Psi(x)\in Q'^{(0)}$ is a solution to multidimensional Dirac--Hestenes equation~\eqref{Dir-Hest}, then $\varphi(x)\in L(t): \varphi(x)=\Psi(x)t$ is the corresponding solution to multidimensional Dirac equation~\eqref{Dirac}.
\end{thm}

\begin{rem}
    We omit the dependence on $x$ for all variables in the multidimensional Dirac—Hestenes equation and the multidemensional Dirac equation to make the proof more concise.
\end{rem}

\begin{proof}
Let us show that $\Psi: \varphi=\Psi t$ is a unique  solution to equation \eqref{Dir-Hest} if $\varphi$ is a solution to equation \eqref{Dirac}.
    First, the uniqueness in the statement follows from Lemma \ref{decomposition}.   
    Equation \eqref{Dirac} is multiplied by $-i$ to make its form similar to equation~\eqref{Dir-Hest}:
    \[
    \sum_{\mu=0}^{2d-1} e^{\mu}(\partial_{\mu}+i  a_{\mu})\varphi+im\varphi=0.
    \]

Using Lemma \ref{decomposition} and property \eqref{it}, we get:
    \[
    \sum_{\mu=0}^{2d-1} e^{\mu} (\partial_{\mu} \Psi t + a_{\mu}\Psi It)+m\Psi It=0.
    \]

   It is possible to factor out $t$ from the parentheses. A wave function $\Psi$ is an even element. However, the terms in the sum are not even elements since $e^{\mu}$ is the odd element. Therefore, we use property \eqref{et} to make the terms even:
    \[
    \left(\sum_{\mu=0}^{2d-1} e^{\mu} (\partial_{\mu} \Psi  + a_{\mu}\Psi I)E+m\Psi I\right)t=0.
    \]

    To apply Lemma \ref{Yt=0}, the element inside the parentheses should belong to the algebra $Q'^{(0)}$, but it belongs to $\cl^{(0)}_{1,2d-1}$. We split the sum over $\mu$ into two parts. The first sum contains only terms that belong to the algebra $Q'^{(0)}$:
    \begin{equation*}
        \begin{aligned}
            &\sum_{\mu = 0,1,2,3,5,7,\ldots,2d-1}e^{\mu} (\partial_{\mu} \Psi+ \Psi a_{\mu} I) E t \\
           & + \sum_{\mu=4,6,\ldots,2d-2}  e^{\mu} (\partial_{\mu} \Psi  + \Psi a_{\mu} I)E t +m\Psi It=0. 
        \end{aligned}
    \end{equation*}

    Let consider the second sum. We use the notation $A_{\mu} = \partial_{\mu} \Psi +  \Psi a_{\mu} I$ to make the proof more concise. Note that $A_{\mu} \in Q'^{(0)}$. 
    In the second sum, we have $e^{\mu} A_{\mu} = A_{\mu}e^{\mu}$  since $A_{\mu}$ does not contain generators with an even index and belongs to the even subalgebra.
    Taking into account anticommutation relations~\eqref{generator}, it becomes clear that $e^{\mu} E = - E e^{\mu}$. Hence, the second sum is transformed into:
    \begin{equation}
    \label{secondsum}
        \sum_{\mu = 4,6,\ldots,2d-2}\!\!\!\!\!\!\!\!\! e^{\mu} A_{\mu} E t = -\!\!\!\!\!\!\!\!\!\sum_{\mu = 4,6,\ldots,2d-2}  \!\!\!\!\!\!\!\!\!A_{\mu}E e^{\mu}t.
    \end{equation}

    The following equality can be easily proved:
    $$e^{2\mu-1}e^{2\mu} t = -i t,\quad \mu = 1,2,\ldots,d-1.$$
Multiplying the equality on the left by $e^{2\mu-1}$ and using property \eqref{it}, we get:
    \begin{equation}
    \label{e2k}
         e^{2\mu} t = e^{2\mu-1} I t.
    \end{equation}

We substitute equality \eqref{e2k} into \eqref{secondsum} and  interchange $e^{\mu-1}$ with $E$. Thus, the second sum becomes:
    $$\sum_{\mu = 4,6,\ldots,2d-2}\!\!\!\!\!\!\!\!\! e^{\mu} A_{\mu} E t = \!\!\!\!\!\!\!\!\!\sum_{\mu = 4,6,\ldots,2d-2}\!\!\!\!\!\!\!\!\!  A_{\mu} e^{\mu-1} E  It.$$

    Shifting the index $\mu$ by $1$, we obtain:
    $$ \sum_{\mu = 4,6,\ldots,2d-2}\!\!\!\!\!\!\!\!\!  A_{\mu}e^{\mu-1} E I t = \!\!\!\!\!\!\!\!\!\sum_{\mu = 3,5,\ldots,2d-3} \!\!\!\!\!\!\!\!\! A_{\mu+1}e^{\mu} E I t,$$
    where $A_{\mu+1}e^{\mu}  E I \in Q'^{(0)}$. Therefore, it is able to get multidimensional Dirac--Hestenes equation \eqref{Dir-Hest} by applying Lemma \ref{Yt=0} to the following equation:
    $$ \left(\sum_{\mu=0,1,2,3,5,7,\ldots,2d-1}\!\!\!\!\!\!\!\!\!\!\!\!\!\!\!e^{\mu} A_{\mu} E+\sum_{\mu=3,5,\ldots,2d-3}\!\!\!\!\!\!\!\!\! A_{\mu+1}e^{\mu} E I+m\Psi I\right)t=0.$$

    The second statement of Theorem \ref{EvenD-H} is able to be proven in the reverse way. We multiply equation \eqref{Dir-Hest} on the right by $t$ and combine the two real sums into the one complex sum using properties \eqref{it}, \eqref{et}, \eqref{e2k}. Then, we replace $\Psi t$ by~$\varphi$.  
\end{proof}

For the case $n=2d$, we should consider one additional term that contains the generator $e^{2d}$ in the sum in multidimensional Dirac equation \eqref{Dirac} compared to the case $n=2d-1$. Since a solution to equation (\ref{Dirac}) can be a semi-spinor or a double spinor in the even case, a form of the multidimensional Dirac--Hestenes equation should differ. 

Let us present the multidimensional Dirac--Hestenes equation in the case of a semi-spinor:
\begin{equation}
\label{SemispinorDir-Hest}
\begin{aligned}
      &\sum_{\mu = 0,1,2,3,5,7,\ldots,2d-1}e^{\mu} (\partial_{\mu} \Psi+ \Psi a_{\mu} I) E\\
      &+\sum_{\mu = 3,5,\ldots,2d-3,2d-1}  (\partial_{\mu+1} \Psi + \Psi a_{\mu+1} I ) e^{\mu} E I +m\Psi I=0,
\end{aligned}
\end{equation}
where $I=-e^{12}$ and $E=e^0$. According to Lemma \ref{decompositionsemi} and relation \eqref{varphisemi},
$\Psi$ commutes with $e^{2d}$. Therefore, the term in (\ref{Dirac}) that contains this generator, has been added to the second sum.

\begin{thm}
\label{OddD-Hsemi}
    Let $t$ be primitive Hermitian idempotent \eqref{oddtsemi} and $Q'$ be the real algebra:
    \begin{equation}
        Q' = \cl(e^0,e^1,e^2,e^3,e^5,e^7,\ldots,e^{2d-1})\subset \cl_{1,2d}.\label{Q2}
    \end{equation}

    If  $\varphi(x)\in L(t)$ is a solution (semi-spinor) to multidimensional Dirac equation (\ref{Dirac}), then $\Psi(x)\in Q'^{(0)}: \varphi(x)=\Psi(x)t$ is the corresponding solution to multidimensional Dirac--Hestenes equation \eqref{SemispinorDir-Hest}.

If $\Psi(x)\in Q'^{(0)}$ is a solution to multidimensional Dirac--Hestenes equation~\eqref{SemispinorDir-Hest}, then $\varphi(x)\in L(t): \varphi(x)=\Psi(x)t$ is the corresponding solution to multidimensional Dirac equation~\eqref{Dirac}.
\end{thm}

\begin{proof}
We can make the same steps as in the proof of Theorem \ref{EvenD-H} to get:
  \begin{equation*}
  \begin{aligned}
        &\sum_{\mu = 0,1,2,3,5,7,\ldots,2d-1}e^{\mu} (\partial_{\mu} \Psi+ \Psi a_{\mu} I) E t\\ &+\sum_{\mu=4,6,\ldots,2d}  e^{\mu} (\partial_{\mu} \Psi  + \Psi a_{\mu} I)E t +m\Psi It=0.
  \end{aligned}
\end{equation*}
The second sum has been obtained to consider terms that contain the generators with even indices, in equation \eqref{Dirac}. It is notable that these generators commute with $\Psi$.

Using the following equality in the second sum:
    \[
    e^{2\mu} t = i e^{2\mu-1} t,\quad \mu = 1,2,\ldots,d
    \]
and Lemma \ref{Yt=0semi}, we get multidimensional Dirac--Hestenes equation \eqref{SemispinorDir-Hest}. It follows from Lemma \ref{decompositionsemi} that the second statement of Theorem \ref{OddD-Hsemi} holds.
\end{proof}

Let us present the multidimensional Dirac--Hestenes equation in the case of a double spinor:
\begin{equation}
\label{DoublespinorDir-Hest}
\begin{aligned}
     &\sum_{\mu = 0,1,2,3,5,7,\ldots,2d-3,2d-1,2d}e^{\mu} (\partial_{\mu} \Psi+ \Psi a_{\mu} I) E\\ &+ \sum_{\mu = 3,5,\ldots,2d-3}  (\partial_{\mu+1} \Psi + \Psi a_{\mu+1} I ) e^{\mu} E I +m\Psi I=0,
\end{aligned}
\end{equation}
where $I=-e^{12}$ and $E=e^0$. According to Lemma \ref{decompositiondouble} and relation \eqref{varphidouble}, $\Psi$ does not commute with the generator $e^{2d}$. Hence, the term in (\ref{Dirac}) that contains this generator, has been added to the first sum.

\begin{thm}
\label{OddD-Hdouble}
    Let $t$ be primitive Hermitian idempotent \eqref{oddtdouble} and $Q'$ be the real algebra:
    \begin{equation}
        Q' = \cl(e^0,e^1,e^2,e^3,e^5,e^7,\ldots,e^{2d-1},e^{2d})\subset \cl_{1,2d}.\label{Q3}
    \end{equation}

    If  $\varphi(x)\in L(t)$ is a solution (double spinor) to multidimensional Dirac equation (\ref{Dirac}), then $\Psi(x)\in Q'^{(0)}: \varphi(x)=\Psi(x)t$ is the corresponding solution to multidimensional Dirac--Hestenes equation \eqref{DoublespinorDir-Hest}.

If $\Psi(x)\in Q'^{(0)}$ is a solution to multidimensional Dirac--Hestenes equation~\eqref{DoublespinorDir-Hest}, then $\varphi(x)\in L(t): \varphi(x)=\Psi(x)t$ is the corresponding solution to multidimensional Dirac equation~\eqref{Dirac}.
\end{thm}

\begin{proof}
The difference between the proofs of Theorem \ref{OddD-Hsemi} and Theorem \ref{OddD-Hdouble} is to consider the term that contains the generator $e^{2d}$. In the case of a double spinor, we include it in the first sum:
 \begin{equation*}
 \begin{aligned}
    &\sum_{\mu = 0,1,2,3,5,7,\ldots,2d-3,2d-1,2d}e^{\mu} (\partial_{\mu} \Psi+ \Psi a_{\mu} I) E t\\ 
    &+\sum_{\mu=4,6,\ldots,2d-2}  e^{\mu} (\partial_{\mu} \Psi  + \Psi a_{\mu} I)E t +m\Psi It=0,
 \end{aligned}
\end{equation*}
since the element $\Psi\in Q'^{(0)}$ does not commute with the generator $e^{2d}$.

Using the following equality in the second sum:
    \[
    e^{2\mu} t = i e^{2\mu-1} t,\quad \mu = 1,2,\ldots,d-1,
    \]
and Lemma \ref{Yt=0double}, we get multidimensional Dirac--Hestenes equation \eqref{DoublespinorDir-Hest}. It follows from Lemma \ref{decompositiondouble} that the second statement of Theorem \ref{OddD-Hdouble} holds.
\end{proof}

\begin{ex}[Dirac--Hestenes equation for graphene]

    Dirac equation \eqref{Dirac} which is used to describe electrons in graphene, is considered in $\mathbb{C}\otimes\cl_{1,2}$. 
    The bilayer graphene consists of two folded monolayers of carbon. We denote them as $1$ and $2$ monolayers. The elementary cell of each monolayer has two inequivalent regions $A$ and $B$. Therefore, the wave function has four real components, which can be arranged in the following order $\{\varphi_{A,1},\varphi_{A,2},\varphi_{B,1},\varphi_{B,2}\}$ \cite{castro2009electronic,dargys2014pseudospin}. In other words, the wave function can be represented as a two-component complex-valued vector. 
    Thus, a wave function is a semi-spinor.

To show the basis decomposition of a wave function $\varphi\in L(t)$, let us fix the idempotent $t$ as \eqref{oddtsemi} in the case $d=1$. 
In accordance with the described method, the algebra $Q$ has only one generator $e^1$, in other words, $Q = \cl(e_1)\subset\mathbb{C}\otimes\cl_{1,2}$. Therefore, the orthonormal basis $\{\tau^{k}\}_{k=1}^2$ of $L(t)$ has form~\eqref{tauF}:
\begin{eqnarray}
    \tau^1 = 2 e t,\quad \tau^2 = 2 e^{01} t
\end{eqnarray}
and the basis decomposition of $\varphi$ is:
\begin{eqnarray}
    \varphi = (\varphi_{A,1}+i\varphi_{A,2})\tau^1 + (\varphi_{B,1}+i\varphi_{B,2})\tau^2.
\end{eqnarray}

Let us introduce the Dirac--Hestenes equation for graphene: 
\begin{equation}
    \sum_{\mu = 0,1,2,3}e^{\mu} (\partial_{\mu} \Psi+ \Psi a_{\mu} I) E + m\Psi I=0,
\end{equation}
where $I=-e^{12}$ and $E=e^0$.

A solution $\Psi$ to the Dirac--Hestenes equation belongs to the even subalgebra of the auxiliary algebra $Q'$. The algebra $Q'$ has been constructed by adding the generators $e^0$ and $e^2$ to the algebra $Q$. Therefore, we get that $Q' = \cl(e^0,e^1,e^2)=\cl_{1,2}$ and $\Psi\in\cl_{1,2}^{(0)}$. Due to Lemma \ref{decompositionsemi}, the explicit form $\Psi$ is:
\begin{eqnarray}
    \Psi = 2\left(\varphi_{A,1}e+\varphi_{A,2}I + \varphi_{B,1}e^{01}+\varphi_{B,2} e^{01}I\right)\in\cl^{(0)}_{1,2}.
\end{eqnarray}

\end{ex}

\section{Gauge Invariance}
\label{secInv}

It is known that the four-dimensional Dirac--Hestenes equation has gauge invariance \cite{lasenby1998gravity,Hestenes}. It means that one can shift an electromagnetic potential without physical consequences.
The multidimensional Dirac--Hestenes equation also has gauge invariance. To prove this fact, we use the following lemma.

\begin{lem}
\label{e012}
Let $I=-e^{12}$ and $E=e^0$ be elements of the geometric algebra $\cl_{1,n}$. Then:
\begin{equation}
\label{exp}
    \exp({I\lambda}) E = E \exp({I\lambda}),\quad \forall\lambda\in\mathbb{R}.
\end{equation}
\end{lem}

\begin{proof}
    We decompose the exponent into the series:
    \[
    \exp(I\lambda) = \sum_{k=0}^{\infty} \frac{(I\lambda)^k}{k!}=e\sum_{k=0}^{\infty} \frac{(-1)^k\lambda^{2k}}{(2k)!}  +I\sum_{k=0}^{\infty} \frac{(-1)^{k}\lambda^{2k+1}}{(2k+1)!}=e\cos \lambda +I\sin\lambda.
    \]

    Both terms commute with $E$.
\end{proof}

\begin{rem}
    Equation \eqref{exp} could have been derived using the isomorphism between matrix algebra and geometric algebra. Since $I^2=-e$ and Euler's identity holds for matrices, equation \eqref{exp} follows directly from these properties.
\end{rem}

\begin{thm}
\label{gaugetheor}
        Let $Q'$ be the real algebra \eqref{Q1} (or \eqref{Q2} and  \eqref{Q3} respectively). If $\Psi(x)\in Q'^{(0)}$ is a solution to multidimensional Dirac--Hestenes equation~\eqref{Dir-Hest} (or \eqref{SemispinorDir-Hest} and \eqref{DoublespinorDir-Hest} respectively) with an electromagnetic vector-potential $\textbf{a}(x)$, then $\Tilde{\Psi}(x)$ is also a solution to the same multidimensional Dirac--Hestenes equation with an electromagnetic vector-potential $\Tilde{\textbf{a}}(x)$:
        \begin{equation}
            \Tilde{\Psi}(x) = \Psi(x) \exp({I\lambda(x)}), \quad\quad \Tilde{a}_{\mu}(x)= a_{\mu}(x) -\partial_{\mu}\lambda(x),
        \end{equation}
    where $\lambda(x): \mathbb{R}^{1,n}\to \mathbb{R}$.
    \end{thm}

\begin{proof} Let us prove the statement for the case $n=2d-1$ with \eqref{Dir-Hest} 
and \eqref{Q1}. In the case $n=2d$, the statements for semi-spinors with \eqref{SemispinorDir-Hest} and \eqref{Q2}  and for double spinors with \eqref{DoublespinorDir-Hest} and \eqref{Q3} are proved similarly.

We substitute the functions $\Tilde{\Psi}$ and $\Tilde{a}_{\mu}$ into the left side of Dirac--Hestenes equation (\ref{Dir-Hest}):
\begin{eqnarray*}
& &\sum_{\mu = 0,1,2,3,5,7,\ldots,2d-1} e^{\mu} (\partial_{\mu} \left(\Psi \exp({I\lambda})\right)+ \Psi\exp({I\lambda}) (a_{\mu} -\partial_{\mu}\lambda) I) E  \nonumber \\
 & &+\sum_{\mu = 3,5,\ldots,2d-3}  (\partial_{\mu+1} \left(\Psi \exp({I\lambda})\right) + \Psi\exp({I\lambda}) (a_{\mu+1}-\partial_{\mu+1}\lambda) I ) e^{\mu} E I\\
 &&+m\Psi\exp({I\lambda}) I.
\end{eqnarray*}

There is a fact that:
$$\partial_{\mu} (\Psi\exp({I\lambda}))= \partial_{\mu} (\Psi)\exp({I\lambda})+\Psi\partial_{\mu}(\lambda)I\exp({I\lambda}),\; \mu = 0,1,\ldots,2d-1.$$

Simplifying the terms and taking into account Lemma \ref{e012}, we obtain:
\begin{eqnarray*}
& &\left(\sum_{\mu = 0,1,2,3,5,7,\ldots,2d-1} e^{\mu} (\partial_{\mu} \Psi+ \Psi a_{\mu} I) E\right) \exp({I\lambda})   \\
 & &+\left(\sum_{\mu = 3,5,\ldots,2d-3}  (\partial_{\mu+1} \Psi + \Psi a_{\mu+1} I ) e^{\mu} E I  +m\Psi I\right)\exp({I\lambda})=0.
\end{eqnarray*}

Therefore, the wave function $\Tilde{\Psi}$ is also a solution to equation (\ref{Dir-Hest}) with the shifted electromagnetic vector-potential $\Tilde{\textbf{a}}$. 
\end{proof}

\section{Conclusion}
\label{secCon}
\vspace{-0.5em}
Multidimensional Dirac spinors are used in supersymmetry theory, for instance, in supersymmetric Yang–Mills theory or M-theory. Actually, it is convenient to consider the spinors in a real geometric algebra. However, taking into account the matrix representation of the complexified geometric algebra $\mathbb{C}\otimes\cl_{1,n}$, it becomes clear that the cases $n=2d-1$ and $n=2d$ differ. Also, in the even case, a Dirac spinor can be represented as a semi-spinor or a double spinor. In this paper, we have considered all cases. 

One of the key advantages of considering multidimensional Dirac--Hestenes is that its solutions belong to a real subalgebra of $\cl_{1,n}$. We have constructed the auxiliary algebra $Q'$ using $e^0$, $e^2$, and the generators with odd indices. However, the algebra $Q'$ also contains the generator $e^{n}$ in the case of a double spinor. 

For all cases, we have proved lemmas on the unique decomposition of an element of the left ideal $L(t)$ into the product of an element of the even real subalgebra $Q'^{(0)}$ and the corresponding idempotent. Using these lemmas, the equivalence between the multidimensional Dirac--Hestenes equation and the multidimensional Dirac equation has been shown. As an application, we have introduced the explicit form of a solution to three-dimensional Dirac--Hestenes equation for graphene. Also, we have shown that the multidimensional Dirac--Hestenes equation has gauge invariance, similar to the four-dimensional case.

% ------------------------------------------------------------------------

\subsection*{Acknowledgment}

The results of this paper were reported at the ENGAGE Workshop within the International Conference Computer Graphics International 2024 (Geneva, Switzerland, CGI 2024). The authors are grateful to the organizers and the participants of this conference for fruitful discussions.

The authors are grateful to the anonymous reviewers for their careful reading of the paper and helpful comments on how to improve the presentation.

The publication was prepared within the framework of the Academic Fund Program at HSE University (grant 24-00-001 Clifford algebras and applications).

\medskip

\noindent{\bf Data availability} Data sharing not applicable to this article as no datasets were generated or analyzed during the current study.

\medskip

\noindent{\bf Declarations}\\
\noindent{\bf Conflict of interest} The authors declare that they have no conflict of interest.

%Paper:
%Balli, S., Chand, S.: Transmission angle in mechanisms. Mech. Mach. Theory 37(2), 175–195 (2002)
%Book:
%Bayro-Corrochano, E.: Geometric Computing: for Wavelet Transforms, Robot Vision, Learning, Control and Action. Springer Publishing Company Inc., Berlin (2010)

\bibliographystyle{spmpsci}
\bibliography{myBibLib} 

% ------------------------------------------------------------------------
\end{document}